\documentclass[11pt]{article}
\usepackage{tikz}
\usetikzlibrary{positioning} 
\usepackage{graphicx}
\usepackage{float}
\usepackage{subfigure}
\usepackage{caption}
\usepackage{multirow}
\usepackage{makecell}
\usepackage{appendix}
\usepackage[figuresright]{rotating}
\usepackage{booktabs}
\usepackage[linesnumbered,ruled]{algorithm2e}
\usepackage[margin=1in]{geometry}
\usepackage{hyperref}
\usepackage{amsfonts}
\usepackage{mathrsfs}
\usepackage{comment}
\usepackage{amsmath}
\usepackage{amssymb}
\usepackage{amsthm}
\usepackage{amscd}
\usepackage{graphicx}
\usepackage{indentfirst}
\usepackage[all]{xy}
\usepackage{titlesec}
\usepackage{enumerate}
\usepackage{bm}
\usepackage{enumitem}
\usepackage{color}
\usepackage{dsfont}
\usepackage{arydshln}
\usepackage{booktabs}
\newtheorem{theorem}{Theorem}[section]
\newtheorem{lemma}[theorem]{Lemma}

\newtheorem{definition}[theorem]{Definition}

\newtheorem{example}[theorem]{Example}

\begin{document}
\title{Unique Decoding of Hyperderivative Reed-Solomon Codes\footnote{The research was supported by the National Natural Science Foundation of China under the Grants 12222113 and 12441105.}} 
\author{Haojie Gu\footnote{Haojie Gu is with the School of Mathematical Sciences, Capital Normal University, Beijing, China, 100048. Email: 2200502051@cnu.edu.cn.},
	\and Jun Zhang\footnote{Jun Zhang is with the School of Mathematical Sciences, Capital Normal University, Beijing, China, 100048. Email: junz@cnu.edu.cn.}
}

\date{}
\maketitle

\begin{abstract}
	
 Error-correcting codes are combinatorial objects designed to cope with the problem of reliable transmission of information on a noisy channel. A fundamental problem in coding theory and practice is to efficiently decode the received word with errors to obtain the transmitted codeword. In this paper, we consider the decoding problem of Hyperderivative Reed-Solomon (HRS) codes with respect to the NRT metric. Specifically, we propose a Welch-Berlekamp algorithm for the unique decoding of NRT HRS codes.

	\begin{flushleft}
	\end{flushleft}
\end{abstract}

\section{Introduction}

Let $\mathbb{F}_{q}$ be a finite field with size $q$ and characteristic $p$. Let $\mathbb{F}_{q}^n$ be the $n$-dimensional vector space over the finite field $\mathbb{F}_{q}$. An error correction code $\mathcal{C}$ is a subset of $\mathbb{F}_{q}^n$. In order to correct the errors that occurred during the transmission, a certain metric is equipped in the space $\mathbb{F}_{q}^n$ according to the type of error of the channel. Two fundamental problems in coding theory are the constructions of optimal codes with respect to bounds of basic parameters and how to decode the received word with noise efficiently. Sometimes one may write a word  in the form of a matrix  in $\mbox{Mat}_{r\times s}(\mathbb{F}_q)$, then there are Hamming metric, rank metric, $P$-metric and so on widely considered in the literature. 

For any vector $ \boldsymbol{x}=(x_1,x_2,\cdots,x_n)\in \mathbb{F}_{q}^n$, the \emph{Hamming weight} $\omega_{H}( \boldsymbol{x})$ of $ \boldsymbol{x}$ is defined to be the number of non-zero coordinates, i.e.,
$\omega_{H}( \boldsymbol{x})=|\left\{i\,|\,1\leqslant i\leqslant n,\,x_i\neq 0\right\}|.$

An $[n,k,d]$-linear code $\mathcal{C}\subseteq \mathbb{F}_{q}^n$ is a $k$-dimensional linear subspace of $\mathbb{F}_{q}^n$ with minimal distance $d=d_{H}(\mathcal{C})$ defined as $$d_{H}(\mathcal{C})=\min\left\{\omega_{H}(c):c\in\mathcal{C}\backslash\{0\}\right\}.$$
For any vector $\boldsymbol{u}\in\mathbb{F}_{q}^n$, the error distance from $\boldsymbol{u}$ to $\mathcal{C}$ is defined as
$$d_{H}(\boldsymbol{u},\mathcal{C})=\min\{d_{H}(\boldsymbol{u},\boldsymbol{v})\,|\,\boldsymbol{v}\in C\},$$
where $d_{H}(\boldsymbol{u},\boldsymbol{v})=|\{i\,|\,u_{i}\neq v_{i},\,1\le i\le n\}|$
is the Hamming distance between vectors $\boldsymbol{u}$ and $\boldsymbol{v}$. The well-known Singleton bound says that $d\leq n-k+1$ for any $[n,k,d]$ linear code $\mathcal{C}$. If $d=n-k+1$, then $\mathcal{C}$ is called a maximum distance separable (MDS) code under the Hamming metric.

\begin{definition}
	Let $\mathcal{A}=\left\{\alpha_{1},\cdots,\alpha_{n}\right\}\subseteq \mathbb{F}_{q}$ be the evaluation set, then the Reed-Solomon code $RS_{k}(\mathcal{A})$ of length $n$ and dimension $k$ is defined as \begin{equation*}
	RS_{k}(\mathcal{A})=\left\{(f(\alpha_{1}),\cdots,f(\alpha_{n})):f(x)\in \mathbb{F}_{q}[x],\deg(f)\leq k-1\right\}.
	\end{equation*}
\end{definition}

One of the most fundamental questions one can ask about Reed-Solomon codes is their decodability. In particular, consider a protocol in which a sender picks $\boldsymbol{c}\in RS_{k}(\mathcal{A})$ and transmits $\boldsymbol{c}$ across a channel, which is received as a corrupted message $\boldsymbol{y}\in \mathbb{F}_{q}^n$.  We seek to understand under what circumstances $\boldsymbol{c}$ can be recovered from $\boldsymbol{y}$. For example, the Welch-Berlekamp decoder~\cite{welch1983error} can efficiently recover $\boldsymbol{c}$ if the Hamming distance between $\boldsymbol{c}$ and $\boldsymbol{y}$ does not exceed the unique decoding radius $\lfloor(d-1)/2\rfloor$. 

Rank metric codes were introduced by Delsarte in 1978~\cite{delsarte1978bilinear} as a set of $m\times n$ matrices over  finite field $\mathbb{F}_{q}$. Let  $\operatorname{Mat}_{m \times n}\left(\mathbb{F}_q\right)$ is the $\mathbb{F}_q$-vector space of $m \times n$ matrices with entries from $\mathbb{F}_q$. For any matrix $A\in\operatorname{Mat}_{m\times n}(\mathbb{F}_{q})$,  the \emph{Rank weight} $\omega_{R}( A)$ of matrix $A$ is the rank of matrix $A$, i.e., $\omega_{R}(A)=\operatorname{rank}(A)$.
Thus, we can define the rank distance between two matrices $A$ and $B$ over $\mathbb{F}_{q}$  as $d_R(A,B) = \omega_{R}(A-B)$. 
The minimum rank-distance of a code $\mathcal{C}\subseteq \operatorname{Mat}_{m\times n}(\mathbb{F}_{q})$ is $$d_R(\mathcal{C})=\min\left\{d_{R}(A,B):A\neq B\in\mathcal{C}\right\}.$$
For a code $\mathcal{C}\subseteq\operatorname{Mat}_{m\times n}(\mathbb{F}_{q})$ with the minimum rank distance $d_{R}(\mathcal{C})\geq d$, the Singleton bound in the rank metric asserts that the number of codewords in $\mathcal{C}$ is upper bounded by $q^{m(n-d+1)}$, see~\cite{gabidulin1985theory}. A code attaining this bound is called a maximal rank distance (MRD) code.

Gabidulin codes~\cite{delsarte1978bilinear,gabidulin1985theory,roth2002maximum}, an important class of MRD codes, can be viewed as the linearized versions of Reed–Solomon codes. They are defined through the evaluation of linearized polynomials. This perspective positions Gabidulin codes as a perfect analogy to Reed-Solomon codes in the rank metric setting. 
\begin{definition}
    For integers $m\geq n\geq k$ and  elements $\alpha_{1},\cdots,\alpha_{n}\in\mathbb{F}_{q^m}$ that are linearly independent over $\mathbb{F}_{q}$, the corresponding Gabidulin codes over $\mathbb{F}_{q^m}$ is defined to be the $[n,k]_{\mathbb{F}_{q^m}}$-linear code
    \begin{equation*}
        \mathcal{G}_{n,k}(\alpha_{1},\cdots,\alpha_{n})=\left\{\left(f(\alpha_{1}),\cdots,f(\alpha_{n})\right):f(X)=\sum\limits_{i=1}^{k}c_{i}X^{q^{i-1}},c_{1},c_{2},\cdots,c_{k}\in\mathbb{F}_{q^m}\right\}\subseteq\mathbb{F}_{q^m}^n.
    \end{equation*}
\end{definition}

Similar to Reed–Solomon codes, one can design highly efficient encoding and unique decoding algorithms for Gabidulin codes by generalizing the Berlekamp-Welch algorithm~\cite{koetter2008coding}.

The  Niederreiter-Rosenbloom-Tsfasman (NRT) codes to be discussed in this paper is also defined in matrix space.
The study of codes for the NRT metric has a long history of parallel developments. A first source comes from a series of papers by Niederreiter ~\cite{niederreiter1987point,niederreiter1991combinatorial,niederreiter1992orthogonal}. In 1997, Rosenbloom and Tsfasman introduced a new metric in the $\mathbb{F}_{q}$-vector space of $r\times s$ matrices with an information-theoretic motivation~\cite{rosenbloom1997codes}. Since then coding-theoretic questions with respect to this metric have been investigated, including MDS codes~\cite{dougherty2004maximum,dougherty2002maximum,skriganov2001coding}, MacWilliams Duality~\cite{dougherty2002macwilliams,ozen2015identity}, structure and decoding of linear codes~\cite{ozen2006linear,panek2010classification,wang2025new}, coverings~\cite{castoldi2015covering} and so on. The study of those metric spaces was also motivated by the study of distributions of points in the unit cube (see for instance~\cite{barg2007bounds1,dougherty2004maximum,skriganov2001coding}). It was eventually realized that this matrix metric corresponds to a $P$-metric where $P$ is the union of $r$ chains of length $s$.

After introducing the NRT metric, a fundamental problem in coding theory is how to maximize the minimum distance after fixing the dimension $k$ and length $n$ with respect to NRT metric. In \cite{rosenbloom1997codes}, Rosenbloom and Tsfasman get the Singleton bound for NRT metric. Interestingly, as an extension of RS codes, Hyperderivative Reed-Solomon (HRS) code can obtain Singleton bound under the NRT metric, which will be introduced next. Surprisingly, HRS code cannot reach the Singleton bound under the Hamming metric. Similar to the RS code under Hamming metric, we naturally care about the decoding problem of HRS code under NRT metric. In this paper,  we especially present the Welch-Berlekamp algorithm, a unique decoding algorithm for HRS codes under the NRT metric.

The organization is as follows. In Section 2, we introduce basic objects, including Hyperderivatives and NRT metrics. In Section 3, for the unique decoding of NRT HRS codes, we extend the Welch-Berlekamp algorithm for RS codes from the Hamming metric to the NRT metric.

\section{Preliminaries and Notation}
\subsection{Hyperderivatives}

Let $\mathcal{P}(t-1)$ represent the space of polynomials $f(X)\in \mathbb{F}_{q}[X]$ with degree at most $t-1$. The following definition may sometimes be referred to as a Hasse derivative, however we will strictly refer to this as a hyperderivative, as used by Skriganov in~\cite{skriganov2001coding}
and also mentioned in Niederreiter and Lidl’s textbook~\cite{lidl1997finite}.

\begin{definition}
	Let $t$ be a positive integer. Let $f(x)\in \mathcal{P}(t-1)$ be given in the form $f(x)=f_0+ f_1x +\cdots+f_{t-1}x^{t-1}$. Let $j<p$, then
	the $j$-th hyperderivative of $f(x)$ is the polynomial defined by 
	\begin{equation}
	\partial^{(j)} f(x)=\binom{0}{j}f_{0}x^{-j}+\binom{1}{j}f_{1}x^{1-j}+\cdots+\binom{t-1}{j}f_{t-1}x^{t-1-j},
	\end{equation}
	where we use the convention that $\binom{i}{j}=\left\{
	\begin{array}{cc}
	\frac{i!}{(i-j)!j!},&\mbox{if}\ 0\leq j\leq i;\\
	0,&\mbox{otherwise.}
	\end{array}\right.$
\end{definition} 

For $\alpha\in \mathbb{F}_{q}$, we denote by $\nu_f (\alpha)$ the
order of vanishing of $f(x)$ at $\alpha$. This means that the highest power of $(x-\alpha)$ that divides $f(x)$ is $\nu_{f}(\alpha)$. In other
words, we have $f(x)=(x-\alpha)^{\nu_{f}(\alpha)}g(x)$
for some polynomial $g(x)$ such that $g(\alpha)\neq 0$.

We mention two direct consequences of this expansion.

$(1)$\  For all $m \in \mathbb{N}$, we have

$$
\nu_f(u)=m \Longleftrightarrow \partial^{(j)} f(u)=0 \text { for all } j \in\{0,1, \ldots, m-1\} \text { but } \partial^{(m)} f(u) \neq 0.
$$

(2) For any integer $a>\deg(f(x))$, we have 
$$
\partial^{(a)} f(x)=0.
$$

\begin{lemma}[\cite{dvir2013extensions}]\label{Lem:Vandermonde identity}
	For any $A(X),B(X)\in \mathbb{F}_{q}[X]$ and $\ell\geq 0$, we have \begin{equation*}
	\partial^{(\ell)}(A(X)B(X))=\sum\limits_{i=0}^{\ell}\partial^{(i)}A(X)\partial^{(\ell-i)}B(X).
	\end{equation*}
\end{lemma}

 Given $\mathcal{A}=\left\{\alpha_{1},\cdots,\alpha_{r}\right\}\subseteq\mathbb{F}_{q}$ and $\boldsymbol{y}=(y_{i,j})_{1\leq i\leq s\atop 1\leq j\leq r}\in\operatorname{Mat}_{s\times r}(\mathbb{F}_{q})$, we define the Hermite interpolation $H(X;\mathcal{A},\boldsymbol{y})$ as the minimal degree polynomial $H(X)$ which satisfies
\begin{equation*}
    \partial^{(i-1)}H(\alpha_{j})=y_{i,j}
\end{equation*}
for all $1\leq i\leq s,1\leq j\leq r$. Note that the Hermite interpolation is an extension of the Lagrange interpolation. Von Zur Gathen and Gerhard~\cite{von2003modern} presented a fast implementation of Hermite interpolation. 

Denote by $M(n)$ the minimal number of arithmetic operations necessary to multiply two polynomials of degree $n$. If the ground field supports fast Fourier transform then $M(n)\leq O(n\log n)$, otherwise $M(n)\leq O(n \log n \log\log n)$~\cite[Theorem~8.18, Theorem~8.22]{von2003modern}.


\begin{theorem}[\cite{von2003modern}]\label{The:2.3}
	Let $r,s$ be positive integers satisfying $s\leq p$, $\mathcal{A}=\left\{\alpha_{1},\cdots,\alpha_{r}\right\}\subseteq \mathbb{F}_{q}$ and $\boldsymbol{y}=(y_{i,j})_{1\leq i\leq s\atop 1\leq j\leq r}\in \operatorname{Mat}_{s\times r}(\mathbb{F}_{q})$. There exists a unique 
    polynomial $H(X)\in\mathcal{P}(rs-1)$ that  satisfies $\partial^{(j-1)}H(\alpha_{i})=y_{j,i}$ for all $1\leq i\leq r,1\leq j\leq s$. Moreover, $H(X;\mathcal{A},\boldsymbol{y})$ can be computed in time $O(M(sr)\log(sr))$. 
\end{theorem}

\subsection{The NRT Metric}

Historically, Brualdi \textit{et al} introduced the concept of a poset metric ($P$-metric) in~\cite{brualdi1995codes}. Further background and development in $P$-metrics is referenced in~\cite{firer2018poset} and~\cite{skriganov2001coding}.  As a special case of $P$-metrics, Niederreiter first introduced NRT-metric in his 1987 paper~\cite{niederreiter1987point}, developing it within the context of uniform distributions of points in Euclidean space. Subsequently, Rosenbloom and Tsfasman independently applied this same metric in their work~\cite{rosenbloom1997codes}, where they referred to it as the $m$-metric and formally introduced its usage into coding theory.






	
	
	

Let $r$ and $s$ be positive integers. 
The NRT-metric on $\operatorname{Mat}_{s \times r}\left(\mathbb{F}_q\right)$ is defined as follows.

Let $A:=\left(a_{i j}\right)_{\substack{i=1, \ldots, s \\ j=1, \ldots, r}} \in \operatorname{Mat}_{s \times r}\left(\mathbb{F}_q\right)$. Let $A_1, \ldots, A_r$ denote the columns of $A$. The NRT-weight of the column $A_j$, denoted $\omega_{N}\left(A_j\right)$, is defined to be 
\begin{equation*}
\left\{
\begin{array}{cc}
s-i+1,&if\  A_{j}\neq \boldsymbol{0} \\
0,&if\ A_{j}=\boldsymbol{0}
\end{array}
\right.,\end{equation*}
where $i$ is the index of the first nonzero entry of $A_j$ from above to betow. Then the NRT-weight of the matrix $A$, denoted by $\omega_{N}(A)$, is defined to be the sum of NRT-weights of all columns of $A$. That is,
$$
\omega_{N}(A)=\omega_{N}\left(A_1\right)+\cdots+\omega_{N}\left(A_r\right).
$$
Next, the NRT-distance of two matrices $A$ and $B$ is defined by
$$d_{N}(A,B)=\omega_{N}(A-B).
$$
In~\cite{brualdi1995codes}, Brualdi \textit{et al} showed that $P$-distance is in fact a metric on $\mathbb{F}_{q}^n$. Thus,
it is easy to check that NRT-distance is a metric on $\operatorname{Mat}_{s\times r}(\mathbb{F}_{q})$. 

A code that uses the $\operatorname{NRT}$-metric (instead of the Hamming) will be referred to as a
NRT Code or $\operatorname{NRT}$-code. Under the $\operatorname{NRT}$ metric there is an analogue for the minimum distance. Let $\mathcal{C}$ be a subset of $\operatorname{Mat}_{s\times r}(\mathbb{F}_{q})$. Denote by $d_{N}(\mathcal{C})$ the minimum nonzero weight $\omega_{N}(V-W)$,
where $V\neq W\in \mathcal{C}$. There is an analogue of the Singleton bound~\cite{hyun2008maximum} for the
$\operatorname{NRT}$ metric as follows:
$$\left|\mathcal{C}\right|\leq q^{n-d_{N}(\mathcal{C})+1}.$$
In particular, if $\mathcal{C}$ is a linear $[rs,t]$ $\operatorname{NRT}$-code, we have $$d_{N}(\mathcal{C})\leq rs-t+1.$$
We will refer to linear $\operatorname{NRT}$-codes as $[rs,t,d_N]$ codes.  If $d_{N}=rs-t+1$, then we call this code an MDS $\operatorname{NRT}$-metric code.

\subsection{Hyperderivative Reed-Solomon Codes}

Suppose positive integers $r,s$, and $t$ satisfy $s\leq p,r\leq q$ and $t\leq rs$. The construction is based on an $r$-tuple $\mathcal{A}= (\alpha_1,\cdots,\alpha_r)\in \mathbb{F}_q^r$, called evaluation points, and an $s\times r$ matrix $V= (v_{ij})_{i=1,\cdots,s\atop j=1,\cdots,r}$ with entries in $\mathbb{F}_{q}^*$, called the multiplier matrix. The Hyperderivative Reed-Solomon (HRS) code, denoted by $HRS(\mathcal{A}, V,t-1)$, is defined as the image of the following evaluation map:
\begin{equation*}
\begin{aligned}
  EV_{\mathcal{A},V}:\mathcal{P}(t-1)&\rightarrow\ \operatorname{Mat}_{s\times r}(\mathbb{F}_{q})\\
f(z)&\mapsto \begin{pmatrix}
v_{1,1}\partial^{(0)}f(\alpha_{1})&v_{1,2}\partial^{(0)}f(\alpha_{2})
&\cdots&v_{1,r}\partial^{(0)}f(\alpha_{r})\\
v_{2,1}\partial^{(1)}f(\alpha_{1})&v_{2,2}\partial^{(1)}f(\alpha_{2})
 &\cdots&v_{2,r}\partial^{(1)}f(\alpha_{r})\\
\vdots&\vdots&\vdots&\vdots\\
v_{s,1}\partial^{(s-1)}f(\alpha_{1})&v_{s,2}\partial^{(s-1)}f(\alpha_{2})
&\cdots&v_{s,r}\partial^{(s-1)}f(\alpha_{r})\\
\end{pmatrix}
\end{aligned},
\end{equation*}
where $\partial^{(i)} f$ stands for the $i$-th hyperderivative of $f$.

When the multiplier matrix $V$ is the all-ones matrix, the definition was first introduced by Rosenbloom and Tsfasman in their seminal paper~\cite{rosenbloom1997codes}. 

Next, we present the formula for the NRT-weight of a codeword $EV_{\mathcal{A},V}(f)\in\operatorname{Mat}_{s\times r}(\mathbb{F}_{q})$ for arbitrary $s\leq p$.

\begin{lemma}[\cite{skriganov2001coding}]\label{the weight of HRS}
	Let $A$ denote the matrix $EV_{\mathcal{A},V}(f)$, where $f\in \mathcal{P}(t-1)$. Let $A_j$ denotes
	the $j$-th column of $A$, then the NRT-weight of $A_{j}$ is given by
	\begin{equation*}
	\omega_{N}(A_{j})=s-\nu_{f}(\alpha_{j}).
	\end{equation*}
	Therefore, the NRT-weight of the matrix $A$ is given by
	\begin{equation*}
	\omega_{N}(EV_{\mathcal{A},V}(f))=\sum\limits_{j=1}^{r}(s-\nu_{f}(\alpha_{j}))=sr-\sum\limits_{j=1}^{r}\nu_{f}(\alpha_{j}).
	\end{equation*}
\end{lemma}

It was shown that HRS codes are MDS codes under the NRT metric.
\begin{theorem}[\cite{skriganov2001coding}]\label{HRS is MDS}
	Let $\mathcal{A}=\left\{\alpha_{1},\cdots,\alpha_{r}\right\}\subseteq\mathbb{F}_{q}$ be a set of $r$ distinct elements of $\mathbb{F}_{q}$ and $V\in (\mathbb{F}_{q}^{*})^{s\times r}$. Then the HRS code $HRS(\mathcal{A}, V,t-1)$ is a $t$-dimensional MDS NRT-metric code.
\end{theorem}

\section{The Welch-Berlekamp Algorithm of NRT HRS codes}

Consider the $[sr,t,d=sr-t+1]_q$ NRT HRS code with evaluation points $\mathcal{A}=\left\{\alpha_{1},\cdots,\alpha_{r}\right\}$ and $s\times r$ matrix $V=\boldsymbol{1}_{s\times r}$.   Let $EV_{\mathcal{A},\boldsymbol{1}}(P(X))$ be the transmitted codeword and  $\boldsymbol{y}=(y_{i,j})_{1\leq i\leq s\atop 1\leq j\leq r}\in \operatorname{Mat}_{s\times r}(\mathbb{F}_{q})$ be the received word, where $P(X)\in\mathcal{P}(t-1)$. The fundamental problem is how to recover the correct codeword $EV_{\mathcal{A},\boldsymbol{1}}(P(X))$ or a codeword $EV_{\mathcal{A},\boldsymbol{1}}(Q(X))$ closest from the received vector $\boldsymbol{y}$. The main result of this paper is to give a decoding algorithm (Algorithm~\ref{Algo:1}) which entends the Welch-Berlekamp Algorithm of RS codes under the Hamming metric to HRS codes under NRT metric and corrects up to $e<\frac{rs-t+1}{2}$ errors in polynomial time. We will call Algorithm~\ref{Algo:1} the Welch-Berlekamp Algorithm of NRT HRS codes. We point out that when $s=1$, our algorithm degenerates the original Welch-Berlekamp Algorithm of RS codes.

\begin{algorithm}
	\small{	\DontPrintSemicolon
		\SetAlgoLined
		\KwIn{$r,s\geq 1,1\leq t\leq rs,0<e<\frac{rs-t+1}{2},r\text{-subset} \ \mathcal{A}=\left\{\alpha_{1},\cdots,\alpha_{r}\right\},\boldsymbol{y}=(y_{i,j})\in \operatorname{Mat}_{s\times r}(\mathbb{F}_{q})$}
		\KwOut{A polynomial $P(X)$ of a degree at most $t-1$ or fail}
		\text { Compute a non-zero polynomial } $E(X)$
        \text { of degree exactly $e$}, \text {and a polynomial $N(X)$}  \text { of degree } \text { at most $e+t-1$ such that for all } $1\leq i\leq r$, we all have
		\begin{equation}\label{Algo1:step 1}
		\left\{
		\begin{array}{rl}
		N(\alpha_{i})=&y_{1,i}E(\alpha_{i})\\
		\partial^{(1)}N(\alpha_{i})=&\sum\limits_{j=1}^{2}y_{j,i}\partial^{(2-j)}E(\alpha_{i})\\
		\vdots\\
		\partial^{(s-1)}N(\alpha_{i})=&\sum\limits_{j=1}^{s}y_{j,i}\partial^{(s-j)}E(\alpha_{i})\\
		\end{array}\right.
		\end{equation}\label{Algo1:Step 1}\;
		\If{ E(X) \text { and } N(X) \text { as above do not exist or } E(X) \text { does not divide } N(X)}{
			\text{Return fail}}
		$P(X)\leftarrow\frac{N(X)}{E(X)}$\label{step:5}\;
		\If{$\omega_{N}(EV_{\mathcal{A},\boldsymbol{1}}(P(X))-\boldsymbol{y})>e$}{\text{Return fail}}
		\ElseIf{$\omega_{N}(EV_{\mathcal{A},\boldsymbol{1}}(P(X))-\boldsymbol{y})\leq e$}{Return $P(X)$}}
	\caption{Welch-Berlekamp Algorithm of NRT HRS codes}\label{Algo:1}
	\end{algorithm}
	
	Before verifying the correctness of the Welch-Berlekamp Algorithm of NRT HRS codes, let us give an example to illustrate the algorithm.

	\begin{example}
		Let $r=4,s=2,t=4,e\leq\lfloor\frac{rs-t}{2}\rfloor=2$. Let $q=7,\mathcal{A}=\left\{1,2,3,4\right\}\subseteq \mathbb{F}_{7}$. Choose a polynomial $P(X)=X^3+3X^2+2X+5\in \mathcal{P}(3)$ which gives the codeword
        \begin{equation*}
		EV_{\mathcal{A},\boldsymbol{1}}(P(X))=\begin{pmatrix}
		P(1)&P(2)&P(3)&P(4)\\ \partial^{(1)}P(1)&\partial^{(1)}P(2)&\partial^{(1)}P(3)&\partial^{(1)}P(4)\\
		\end{pmatrix}=\begin{pmatrix}
		4&1&2&6\\4&5&5&4
		\end{pmatrix}.
		\end{equation*}
        Suppose the following error with NRT-weight $2$ occured
        $$\begin{pmatrix}
		0&0&0&0\\1&0&1&0
		\end{pmatrix}.$$ So the received word is 
        $$\boldsymbol{y}=\begin{pmatrix}
		4&1&2&6\\
		5&5&6&4
		\end{pmatrix}.$$ 
		Let $N(X)=a_{0}+a_{1}X+a_{2}X^2+a_{3}X^2+a_{4}X^4+a_{5}X^5$ and $E(X)=b_{0}+b_{1}X+X^2$. From Equation~(\ref{Algo1:step 1}) of Algorithm~\ref{Algo:1}, we have
		$$\left\{
		\begin{array}{ll}
		\sum\limits_{j=0}^{5}a_{j}\alpha_{i}^j-y_{1,i}b_{0}-y_{1,i}\alpha_{i}b_{1}=y_{1,i}\alpha_{i}^2,&i=1,2,3,4\\
		\sum\limits_{j=1}^{5}ja_{j}\alpha_{i}^{j-1}-y_{2,i}b_{0}-(y_{1,i}+y_{2,i}\alpha_{i})b_{1}=2y_{1,i}\alpha_{i}+y_{2,i}\alpha_{i}^2,&i=1,2,3,4
		\end{array}\right..$$
		In other words, we have the following system of linear equations: $$M\boldsymbol{x}=A,$$ 
		where $\boldsymbol{x}=(a_{0},a_{1},\cdots,a_{5},b_{0},b_{1})^T,A=(4,4,4,5,6,3,3,0)^T$ and
		\begin{equation*}
		M=\begin{pmatrix}
		1&1&1&1&1&1&3&3\\
		1&2&4&1&2&4&6&5\\
		1&3&2&6&4&5&5&1\\
		1&4&2&1&4&2&1&4\\
		0&1&2&3&4&5&2&5\\
		0&1&4&5&4&3&2&3\\
		0&1&6&6&3&6&1&1\\
		0&1&1&6&4&6&3&6\\
		\end{pmatrix}
		\end{equation*}
		By sovling the system of equations, we obtain the solution $\boldsymbol{x}=(1,0,6,0,6,1,3,3)^T$. Thus, $$N(X)=X^5+6X^4+6X^2+1\quad \mbox{and} \quad E(X)=X^2+3X+3.$$ It is easy to check that $E(X)\mid N(X)$. And hence, Algorithm~\ref{Algo:1} returns
        $$P(X)=\frac{N(X)}{E(X)}=X^3+3X^2+2X+5.$$ Finally, we obtain the transmitted codeword
		\begin{equation*}
		EV_{\mathcal{A},\boldsymbol{1}}(P(X))=\begin{pmatrix}
		P(1)&P(2)&P(3)&P(4)\\ \partial^{(1)}P(1)&\partial^{(1)}P(2)&\partial^{(1)}P(3)&\partial^{(1)}P(4)\\
		\end{pmatrix}=\begin{pmatrix}
		4&1&2&6\\4&5&5&4
		\end{pmatrix}.
		\end{equation*}
        So Algorithm~\ref{Algo:1} really recovers the correct codeword.
        
	\end{example}

    Next we prove the correctness of the proposed Welch-Berlekamp Algorithm of NRT HRS codes.
    
    \begin{lemma}\label{Lem: exist W-B algorithm}
    	There exist a pair of polynomials $E_{1}(X)$ and $N_{1}(X)$ that satisfy Step~\ref{Algo1:Step 1} such that $N_{1}(X)=E_{1}(X)P(X)$.
    \end{lemma}
    \begin{proof}
    By Theorem~\ref{The:2.3}, there exists a  polynomial $H(X)\in\mathcal{P}(rs-1)$ such that $H(X)=H(X;\mathcal{A},\boldsymbol{y})$. Let $Q(X)=P(X)-H(X)$. Denote
    	\begin{equation}\label{equation:E(x)}
    	E_{1}(X)=X^{e-\Delta}\prod\limits_{1\leq i\leq r}(x-\alpha_{i})^{s-\nu_{Q}(\alpha_{i})},
    	\end{equation}
        where $\Delta=\omega_{N}(EV_{\mathcal{A},\boldsymbol{1}}(Q(X)))$. Let $N_{1}(X)=E_{1}(X)P(X)$, then $\deg(E_{1}(X))=e$ and $$\deg(N_{1}(X))=\deg(E_{1}(X))+\deg(P(X))\leq e+t-1.$$
    	
    	For $1\leq i\leq r$, if $1\leq j\leq \nu_{Q}(\alpha_{i})$, then $(X-\alpha_{i})^j|Q(X)$. So $$\partial^{(j-1)}(P(\alpha_{i}))=\partial^{(j-1)}H(\alpha_{i})=y_{j,i}.$$ 
    If $\nu_{Q}(\alpha_{i})<j\leq s$, then $s-j<s-\nu_{Q}(\alpha_{i})$. Thus, $$\partial^{(s-j)}E_{1}(\alpha_{i})=0.$$
    In summary, for all $1\leq i\leq r,1\leq j\leq s$, we have $$(y_{j,i}-\partial^{(j-1)}P(\alpha_{i}))\cdot \partial^{(s-j)}E_{1}(\alpha_{i})=0.$$ 
    So
    $$y_{j,i}\cdot \partial^{(s-j)}E_{1}(\alpha_{i})=\partial^{(j-1)}P(\alpha_{i})\cdot \partial^{(s-j)}E_{1}(\alpha_{i})$$ 
    and 
    $$\sum\limits_{j=1}^{s}y_{j,i}\partial^{(s-j)}E_{1}(\alpha_{i})=\sum\limits_{j=1}^{s}\partial^{(j-1)}P(\alpha_{i})\cdot\partial^{(s-j)}E_{1}(\alpha_{i}).$$
    By Lemma~\ref{Lem:Vandermonde identity}, we have 
    $$
    \sum\limits_{j=1}^{s}y_{j,i}\partial^{(s-j)}E_{1}(\alpha_{i})=\partial^{(s-1)}(P(\alpha_{i})\cdot E_{1}(\alpha_{i}))=\partial^{(s-1)}N_{1}(\alpha_{i}).
    $$
    	Similarly, for all $1\leq\ell\leq s-1,1\leq i\leq r,1\leq j\leq s-\ell$, we have
    	$$(y_{j,i}-\partial^{(j-1)}P(\alpha_{i}))\cdot \partial^{(s-\ell-j)}E_{1}(\alpha_{i})=0.$$ 
    	 Thus, \begin{equation*}
    	\begin{aligned}
    	\sum\limits_{j=1}^{s-\ell}y_{j,i}\partial^{(s-\ell-j)}E_{1}(\alpha_{i})&=\sum\limits_{j=1}^{s-\ell}\partial^{(j-1)}P(\alpha_{i})\cdot\partial^{(s-\ell-j)}E_{1}(\alpha_{i})=\partial^{(s-\ell-1)}N_{1}(\alpha_{i}).
    	\end{aligned}
    	\end{equation*} 
    \end{proof}
    
    Note that now it suffices to argue that$\frac{N_{1}(X)}{E_{1}(X)}=\frac{N_{2}(X)}{E_{2}(X)}$ for any pair of solutions $(N_1(X),E_1(X))$ and $(N_2(X),E_2(X))$ that satisfy Step~\ref{Algo1:Step 1}. Then by Lemma~\ref{Lem: exist W-B algorithm}, the ratio must be $P(X)$. 
    
    \begin{lemma}\label{Lem:unique}
    	For any two distinct solutions $(E_1(X),N_1(X))\neq (E_2(X),N_2(X))$ in Step~\ref{Algo1:Step 1}, they satisfy $\frac{N_{1}(X)}{E_{1}(X)}=\frac{N_{2}(X)}{E_{2}(X)}$.
    \end{lemma}
 \begin{proof}
 By Theorem~\ref{The:2.3}, there exists a  polynomial $H(X)\in\mathcal{P}(rs-1)$ such that $H(X)=H(X;\mathcal{A},\boldsymbol{y})$. Since $(N_1(X),E_{1}(X)$ and $(N_{2}(X),E_{2}(X))$ satisfy Step~\ref{Algo1:Step 1}, we have 
 	\begin{equation*}
 	\partial^{(\ell-1)}N_{k}(\alpha_{i})=\sum\limits_{j=1}^{\ell}y_{j,i}\partial^{(\ell-j)}E_{k}(\alpha_{i})
 	\end{equation*}
    for all $1\leq \ell\leq s,1\leq i\leq r$ and $k=1,2$.
 	
    Let $F(X)=N_{1}(X)E_{2}(X)-N_{2}(X)E_{1}(X)$. For $1\leq i\leq r$ and $1\leq j\leq s$, we have \begin{equation*}
 	\begin{aligned}
 	\partial^{(j-1)}F(\alpha_{i})&=\partial^{(j-1)}(N_{1}(\alpha_{i})E_{2}(\alpha_{i})-N_{2}(\alpha_{i})E_{1}(\alpha_{i}))\\
 	&=\sum\limits_{\ell=1}^{j}\partial^{(\ell-1)}N_{1}(\alpha_{i})\partial^{(j-\ell)}E_{2}(\alpha_{i})-
 	\sum\limits_{\ell=1}^{j}\partial^{(\ell-1)}N_{2}(\alpha_{i})\partial^{(j-\ell)}E_{1}(\alpha_{i})\\
 	&=\sum\limits_{\ell=1}^{j}\sum\limits_{k=1}^{\ell}y_{k,i}\partial^{(\ell-k)}E_{1}(\alpha_{i})\partial^{(j-\ell)}E_{2}(\alpha_{i})-
 	\sum\limits_{\ell=1}^{j}\sum\limits_{k=1}^{\ell}y_{k,i}\partial^{(\ell-k)}E_{2}(\alpha_{i})\partial^{(j-\ell)}E_{1}(\alpha_{i})\\
 	&=\sum\limits_{k=1}^{j}y_{k,i}\sum\limits_{\ell=k}^{j}\partial^{(\ell-k)}E_{1}(\alpha_{i})\partial^{(j-\ell)}E_{2}(\alpha_{i})-
 	\sum\limits_{k=1}^{j}y_{k,i}\sum\limits_{\ell=k}^{j}\partial^{(\ell-k)}E_{2}(\alpha_{i})\partial^{(j-\ell)}E_{1}(\alpha_{i})\\
 	&=\sum\limits_{k=1}^{j}y_{k,i}\partial^{(j-k)}\left(E_{1}(\alpha_{i})E_{2}(\alpha_{i})\right)-\sum\limits_{k=1}^{j}y_{k,i}\partial^{(j-k)}\left(E_{1}(\alpha_{i})E_{2}(\alpha_{i})\right)=0.
 	\end{aligned}.
 	\end{equation*}
 So $$\prod\limits_{1\leq i\leq r}(X-\alpha_{i})^{s}\mid F(X).$$ But the degree of $F(X)$ satisfies $\deg(F(X))\leq e+e+t-1<rs$. Therefore, $F(X)=0$ is the zero polynomial. Hence, $$\frac{N_{1}(X)}{E_{1}(X)}=\frac{N_{2}(X)}{E_{2}(X)}.$$
 \end{proof}
 From Lemma~\ref{Lem: exist W-B algorithm} and Lemma~\ref{Lem:unique}, we know that if Algorithm~\ref{Algo:1} does not output ``fail", then the algorithm produces a correct output. 
 
 Now, we analyze the running time of Algorithm~\ref{Algo:1}. In Step~\ref{Algo1:Step 1}, $N(X)$ has $e+t$ unknowns and $E(X)$ has $e+1$ unknowns. For each $1\leq i\leq r$, the constraint in (\ref{Algo1:step 1}) is a linear system of equations with $r$ equations in these unknowns. Thus, we have a system of $rs$ linear equations in $2e+t+1<rs+2$ unknowns. By Lemma~\ref{Lem: exist W-B algorithm}, this system of equations has a solution. The only additional requirement is that the degree of the polynomial $E(X)$ be exactly $e$. We have already shown $E(X)$ in equation~(\ref{equation:E(x)}) to satisfy this requirement. So we add a constraint that the coefficient of $X^e$ in $E(X)$ is $1$. Therefore, we have $rs+1$ linear equations in at most $rs+1$ variables, which we can solve in time $O(r^3s^3)$, e.g. by Gaussian elimination.
Finally, note that Step~\ref{step:5} can be implemented in time $O(r^3s^3)$ by “long division”. Thus, we obtain the following main Theorem.
    
    \begin{theorem}\label{unique theorem}
    If $EV_{\mathcal{A},\boldsymbol{1}}(P(X))\in \operatorname{Mat}_{s\times r}(\mathbb{F}_{q})$ is transmitted and at most $e\leq \frac{rs-t}{2}$ errors occur in the received word $\boldsymbol{y}\in \operatorname{Mat}_{s\times r}(\mathbb{F}_{q})$, i.e. $\omega_{C(s,r)}(\boldsymbol{y}-EV_{\mathcal{A},\boldsymbol{1}}(P(X)))\leq e$, where $P(X)$ is a polynomial of degree at most $t-1$, then the Welch-Berlekamp algorithm outputs $P(X )$. Moreover, the Welch-Berlekamp algorithm can be done in time $O(r^3s^3)$.  
    \end{theorem}

\section{Conclusion}
In this paper, we study the decoding problem of NRT HRS codes. We propose a Welch-Berkamp Algorithm for the unique decoding of NRT HRS codes which can correct errors with NRT metric up to the unique decoding radius in polynomial time. Finally, we point out that the Welch-Berkamp Algorithm proposed in this paper can be extended to a list-decoder for NRT HRS codes~\cite{GZ25}.

\bibliographystyle{plain}
\bibliography{HRS-JAAG}
\end{document}